\documentclass[12pt]{amsart}
\theoremstyle{definition}
\newtheorem{definition}{Definition}
\newtheorem{theorem}{Theorem}[section]

\usepackage{url,upref,amscd,amsmath,amsthm}
\usepackage{times}
\usepackage{fancyhdr}
\usepackage[psamsfonts]{amssymb}
\usepackage{cite}
\begin{document}
\title[Generalization ... ElGamal cryptosystem]{A simple generalization of the
  {E}l{G}amal cryptosystem to non-abelian groups II}
\author{Ayan Mahalanobis}\address{Indian Institute of Science
  Education and Research Pune, Pashan Pune-411021, India}
\today
\email{ayanm@iiserpune.ac.in}
\begin{abstract}
This is a study of  the MOR cryptosystem using the special linear
group over finite fields. The automorphism group of the
special linear group is analyzed for this purpose.
At our current state of knowledge, I show
that this MOR cryptosystem has better security than the ElGamal
cryptosystem over finite fields. 
\end{abstract}
\maketitle
\section{Introduction}
The MOR cryptosystem is a generalization of the ElGamal cryptosystem,
where the \emph{discrete logarithm problem} works in the automorphism group
of a group $G$, instead of the group $G$ itself. 
The framework for the
MOR cryptosystem was first proposed in Crypto2001 by Paeng et
al.~\cite{crypto2001}. Mahalanobis~\cite{ayan1} used the group of 
unitriangular matrices for the MOR cryptosystem. That effort was
successful, the MOR cryptosystem over the group of unitriangular
matrices over $\mathbb{F}_q$ is as secure as the ElGamal cryptosystem
over the finite field $\mathbb{F}_q$.  

In this paper we study the
MOR cryptosystem over SL$(d,q)$. If we assume, that the only way to
break the proposed MOR cryptosystem, is to solve the discrete logarithm
problem in the automorphism group; then it follows that the proposed
MOR cryptosystem is \textbf{as secure as} the ElGamal
cryptosystem over $\mathbb{F}_{q^d}$. 

This is a major improvement.
This MOR cryptosystem works with matrices of
degree $d$ over $\mathbb{F}_q$. To encrypt(decrypt) a
plaintext(ciphertext) one works over the field $\mathbb{F}_q$. To
break this cryptosystem, one has to solve a discrete logarithm problem
in $\mathbb{F}_{q^{d}}$. Even for a small positive integer $d$, this provides
us with a \textbf{considerable security} advantage. 

There are some challenges in the implementation of this
cryptosystem. Implementing 
matrix multiplication is hard. Though we have not reached the optimum
speed for that~\cite{don}, it might always stay harder than
multiplication in a finite field. So one needs to find an optimum
strategy to present the automorphisms and the underlying group for the
MOR cryptosystem, see Section 8 for more details. The key-size for
this MOR cryptosystem is big, compared with the conventional ElGamal cryptoystem.

At the end, I provide parameters for the proposed MOR
cryptosystem. I suspect that the parameters are too conservative and the
degree of the matrix is unnecessarily big. The overly conservative estimates are
to show that for 
chosen parameters, the MOR cryptosystem is almost as secure as the
ElGamal cryptosystem over elliptic curves using fields of same size; the golden standard in
public key cryptography. For most practical purposes,
the degree of the matrix can be chosen smaller. However the key-size
for this MOR cryptosystem is larger than that of the ElGamal over
elliptic curves.

\section{The MOR cryptosystem}
This section contains a bare-bone description of the MOR cryptosystem~\cite{crypto2001}, see also
\cite{paeng1}. A description and a critical analysis of the MOR cryptosystem is also in
\cite{ayan1} and the references there. 
\subsection{Description of the MOR cryptosystem}\label{MOR}
Let $G=\langle g_1,g_2,\ldots,g_\tau\rangle$, $\tau\in\mathbb{N}$
be a finite group and $\phi$ a non-trivial (public) automorphism of 
$G$. Alice's keys are as follows: 
\begin{description}\label{keyex}
\item[Private Key] $m$, $m\in\mathbb{N}$.
\item[Public Key] $\left\{\phi(g_i)\right\}_{i=1}^\tau$ and $\left\{\phi^m(g_i)\right\}_{i=1}^\tau$.
\end{description}
\paragraph{\textbf{Encryption}}
\begin{description}
\item[a] To send a message (plaintext) $a\in G$ Bob computes $\phi^r$
  and $\phi^{mr}$ for a random $r\in\mathbb{N}$.
\item[b] The ciphertext is $\left(\left\{\phi^r(g_i)\right\}_{i=1}^\tau,\phi^{mr}(a)\right)$.
\end{description}
\paragraph{\textbf{Decryption}}
\begin{description}
\item[a] Alice knows $m$, so if she receives the ciphertext
  $\left(\phi^r,\phi^{mr}(a)\right)$, she computes $\phi^{mr}$ from $\phi^r$ and
  then $\phi^{-mr}$ and then computes $a$ from $\phi^{mr}(a)$.
\end{description}
If Alice has the information
necessary to find the order of the automorphism $\phi$, then she can use
the identity $\phi^{t-1}=\phi^{-1}$ whenever $\phi^t=1$ to compute  $\phi^{-mr}$ . Also, she can
find the order of some subgroup in which $\phi$ belongs and
use the same identity. However, the smaller the subgroup, more efficient
the decryption algorithm.

\section{The unimodular group of degree $d$ over $\mathbb{F}_q$}
The  group SL$(d,q)$ is the set of all 
matrices of degree $d$ with determinant $1$. It is well known that
SL$(d,q)$ is a normal subgroup of GL$(d,q)$ the group of non-singular matrices of
degree $d$ over $\mathbb{F}_q$. In this article I consider
$\mathbb{F}_q$ to be a finite extension of the prime field
$\mathbb{Z}_p$ of degree $\gamma$ where $\gamma\geq 1$.
\begin{definition}
For distinct ordered pair $(i,j)$, define a matrix unit $e_{i,j}$ as
a matrix of degree $d$, such that, all 
entries in $e_{i,j}$ are 0, except the intersection of the i\textsuperscript{th} row 
and the j\textsuperscript{th} column; which is 1 (the identity in the
field $\mathbb{F}_q$). Matrices of the form
$1+\lambda e_{i,j}$, $\lambda\in\mathbb{F}_q^\times$ and $i\neq j$ 
are called the elementary matrices or elementary transvections. Here
$1$ is the identity matrix of degree $d$. I shall abuse the notation a
little bit and use $1$ for the identity of the field and the matrix
group simultaneously.
\end{definition}
It is known that the group SL$(d,q)$ is generated by elementary transvections~\cite[Theorem 8.8]{rotman}. The fundamental relations between the elementary
transvections are the relations in the field and the ones stated below:
\begin{align}\label{slrelations1}
&[1+\lambda e_{i,j},1+\mu e_{k,l}]=
\left\{
\begin{array}{ccc}
1+\lambda\mu e_{i,l} &\text{if}&j=k,\;\; i\neq l\\
1-\lambda\mu e_{k,j} &\text{if}&i=l,\;\; j\neq k\\ 
1          &\text{otherwise}&
\end{array}\right.
\end{align}
\begin{align}\label{slrelations2}
&\left(1+\lambda e_{i,j}\right)\left(1+\mu
  e_{i,j}\right)=1+\left(\lambda+\mu\right)e_{i,j}\\ 
&\left(1+\lambda e_{i,j}\right)^{-1}=\left(1-\lambda e_{i,j}\right)\\
&\left(1+\lambda e_{i,j}\right)^{k}=1+k\lambda e_{i,j}\;\;\;k\in\mathbb{N}
\end{align}
where $\lambda,\mu\in\mathbb{F}_q$.

\section{Automorphisms of the unimodular group over $\mathbb{F}_q$}
It is well known that the automorphisms of SL$(d,q)$ are generated by the following~\cite{carter,jean,steinberg}:
\begin{description}
\item[Diagonal Automorphism] This is conjugation by a non-scalar
  diagonal matrix. Notice that: since all diagonal matrices are not
  of determinant 1, the diagonal matrices are often not in
  SL$(d,q)$. So a diagonal automorphism is not always an inner automorphism.
\item[Inner Automorphism] This is the most well known automorphism of a
  non-abelian group $G$, defined by $x\mapsto g^{-1}xg$ for
  $g\in G$.
\item[Graph Automorphism] The graph automorphism induces the map
  $A\mapsto \left(A^{-1}\right)^T$, $A\in \text{SL}(d,q)$. Clearly, graph
  automorphisms are involutions, i.e., of order two and are not inner
  automorphisms. 
\item[Field Automorphism] This automorphisms is the action of a field
  automorphism of the underlying field to the individual entries of a
  matrix.  
\end{description}
In this section, I am interested in a special class of inner
automorphisms, ``the permutation automorphisms''. For a permutation
automorphism the conjugator $g$ in the inner automorphism
is a permutation matrix. It is well known that for a 
permutation matrix $P$, det$(P)=\pm 1$ and $P^{-1}=P^T$. The
permutation matrix is constructed by taking the identity matrix $1$
and then exchanging the rows based on some permutation $\alpha$. If the
permutation $\alpha$ is even then the determinant of $P$ is $1$ otherwise
it is $-1$. Note that if the determinant is $-1$, then conjugation by
that permutation matrix is not an inner automorphism; but it is close
to being one and I will treat it like an inner automorphism in this paper. 
\subsubsection{Effect of a permutation automorphism on an
  elementary transvections} 
If $A$ is an elementary transvection, i.e., $A=1+\lambda e_{i,j}$ and
$P$ be a permutation matrix, then $P^{-1}AP=1+\lambda
e_{\alpha^{-1}(i),\alpha^{-1}(j)}$. 
\subsubsection{Effect of a diagonal automorphism on an elementary
  transvection} Let $D=[w_1,w_2,\ldots,w_d]$ be a diagonal
matrix. If $A=1+\lambda e_{i,j}$ then $D^{-1}AD=1+(w_i^{-1}\lambda
w_j)e_{i,j}$. Let us fix a $(i,j)$ such that $1\leq i,j\leq d$, then
look at the \emph{root subgroup} $\langle 
1+\lambda e_{i,j}\rangle$, $\lambda\in\mathbb{F}_q$ and $i\neq j$. This subgroup is
clearly isomorphic to $\mathbb{F}_q^+$. 

Assume for a moment that I am using the MOR cryptosystem as
described in Section \ref{MOR} with $G$ as the root subgroup defined
above and $\phi$ as a diagonal automorphism. Then clearly for some $k\in\mathbb{F}_q^\times$.
\begin{eqnarray*}
\phi:&1+e_{i,j}\mapsto &1+ke_{i,j}\\
\phi^m:&1+e_{i,j}\mapsto &1+k^me_{i,j}. 
\end{eqnarray*}
Clearly we see that this MOR cryptosystem is equivalent to the ElGamal
cryptosystem over finite fields. Since SL$(d,q)$ is generated by
elementary transvections, I claim that using the diagonal
automorphisms of the special linear groups over finite fields, the
MOR cryptosystem is identical to the ElGamal cryptosystem over finite
fields. It is reasonable to assume that there are other automorphisms,
composition of which with the diagonal automorphisms will provide us
with better security.  
\subsubsection{The effect of the graph automorphism on an elementary
  transvection } It is easy to see from the definition of the graph
automorphism that if $A=1+\lambda e_{i,j}$ then
$\left(A^{-1}\right)^T=1-\lambda e_{j,i}$.
\subsubsection{The effect of field automorphisms on an elementary
  transvections}
It is well known that the field automorphisms form a cyclic group
generated by the Frobenius automorphism of the field $\mathbb{F}_q$, given by
$\lambda\mapsto\lambda^p$, where $p$ is the characteristic of the
field $\mathbb{F}_q$. Then the action of field automorphism on an
elementary transvection is $1+\lambda e_{i,j}\mapsto
1+\lambda^{p^s}e_{i,j}$ where $1\leq s<\gamma$.
\section{MOR with monomial automorphisms}
Assume for a moment that I am only using the composition of a diagonal
and an inner automorphism of SL$(d,q)$, i.e., I am using MOR (Section \ref{MOR})
where $\phi=\phi_1\circ\phi_2$ where $\phi_1$ is a diagonal
automorphism and $\phi_2$ is a permutation automorphism. Then clearly
$\phi$ is a monomial automorphism, conjugation by a monomial matrix.
The diagonal automorphism $\phi_1$ changes
$1+e_{i,j}$ to $1+\lambda_{i,j} e_{i,j}$ for some
$\lambda_{i,j}\in\mathbb{F}_q^\times$. Note that the $\lambda_{i,j}$
depends on the diagonal automorphism and once the diagonal
automorphism is fixed $\lambda_{i,j}$ is also fixed for a particular
$(i,j)$. The permutation automorphism $\phi_2$ changes 
$1+\lambda_{i,j} e_{i,j}$ to $1+\lambda_{i,j} e_{\beta(i),\beta(j)}$ where
$\beta=\alpha^{-1}$. Here $\alpha$ is the permutation that gives rise to
the permutation matrix $P$, used in the permutation automorphism.

I now look at the action of the exponentiation of the automorphism
$\phi=\phi_1\circ\phi_2$ on the elementary transvection
$1+e_{i,j}$. Notice that if 
\begin{equation}
\phi:
\begin{CD}
1+e_{i,j} @>\text{diagonal}>> 1+\lambda_{i,j}e_{i,j}
@>\text{permutation}>> 1+\lambda_{i,j}e_{\beta(i),\beta(j)}, 
\end{CD}
\end{equation} then 
\begin{equation}
\phi^m:
\begin{CD}
1+e_{i,j} @>>> 1+\prod\limits_{l=1}^m \lambda_{\beta^{l}(i)\beta^l(j)}e_{\beta^{m}(i),\beta^{m}(j)}
\end{CD}\end{equation}
Now let us assume that the order of $\beta$, $\circ(\beta)=\nu$ then 
\[\phi^\nu : 1+e_{i,j}\mapsto 1+\prod\limits_{l=1}^\nu
\lambda_{\beta^{l}(i)\beta^l(j)}e_{i,j}.\]
This shows that a cycle is formed and if $\nu<m$, then
this reduces the discrete logarithm problem in $\langle
\phi\rangle$ to a discrete logarithm problem in the finite field
$\mathbb{F}_q$. Though it is well known that in the symmetric group 
$S_n$, acting on $n$ points, one can get elements with very high
order. In our problem I am actually interested in the length of the
orbit formed by the action of a cyclic subgroup of $S_n$, generated by
$\beta$, on the set of distinct 
ordered pair of $\{1,2,\ldots,n\}$. It is known that these orbits are
quite small.

Since the permutation $\beta$ is easy to find from the public
information $\phi$ and $\phi^m$, unless the degree of the matrix
$d$ is astronomically big, we do not have any chance for a MOR
cryptosystem which is more secure than that of the ElGamal
cryptosystem over finite fields.

Since the conjugacy problem is easy in GL$(d,q)$, from the public
information of $\phi_1$ and $\phi_2$ one can compute the conjugator
monomial matrices for $\phi_1$ and $\phi_2$ modulo an element of the
center of GL$(d,q)$. I shall come back to this topic later (Section 7.2) in more details.    
\section{Structure of the automorphism group of SL$(d,q)$}
Let us start with a well known theorem describing the structure of the automorphism
group of SL$(d,q)$. Let $\mathcal{A}$ be the group of automorphisms
generated by the diagonal and the inner automorphisms and $\mathcal{B}$
be the group generated by the graph and the field automorphisms. Recall
that the center of the group GL$(d,q)$ is the set of all scalar
matrices $\lambda 1$ where $\lambda\in\mathbb{F}_q^\times$ and $1$ is
the identity matrix of degree $d$. I shall denote the center of
GL$(d,q)$ by $Z$ and \emph{the projective general linear group}
$\dfrac{\text{GL}(d,q)}{Z}$ by PGL$(d,q)$. 

A brief \emph{warning} about the notation. To increase readability
of the text, from now on, the image of $a$ under $f$ will be denoted
either by $a^f$ or by $f(a)$. Also, I denote the conjugation of
$X$ by $A$ as $X^A$. 

\begin{theorem}
The group $\mathcal{A}$ is isomorphic to PGL$(d,q)$ and
Aut$\left(\text{SL}(d,q)\right)$ is a semidirect product of
$\mathcal{A}$ with $\mathcal{B}$. 
\end{theorem}
\begin{proof}
From \cite[Theorem 2.12]{alperin} we know that any element in
GL$(d,q)$ is generated by the set consisting of all invertible
diagonal matrices and all transvections. Then we can define a map
$\digamma:\text{GL}(d,q)\rightarrow\mathcal{A}$ defined by
$\digamma(A)$ maps $X\mapsto X^A$, clearly $\digamma$ is an epimorphism and
Ker$(\digamma)=Z$. From first isomorphism theorem we have that
PGL$(d,q)\cong\mathcal{A}$.

We are left to show that
Aut$\left(\text{SL}(d,q)\right)$ is a semidirect product of
$\mathcal{A}$ with $\mathcal{B}$. To prove this we need to show
that $\mathcal{A}$ is a normal subgroup of
$\text{Aut}\left(\text{SL}(d,q)\right)$ and 
Aut$\left(\text{SL}(d,q)\right)=\mathcal{A}\mathcal{B}$. Notice that
any $f\in\mathcal{B}$ is an automorphism of GL$(d,q)$. With
this in mind we see that for $A\in\text{GL}(d,q)$ and
$X\in\text{SL}(d,q)$
\[X^{fAf^{-1}}=f\left(A^{-1}f^{-1}(X)A\right)=f(A)^{-1}Xf(A)=X^{f(A)}.\] 
This proves that $\mathcal{A}$ is a normal subgroup of
Aut$\left(\text{SL}(d,q)\right)$.    
Now notice that for any $f\in\mathcal{B}$,
$A^{-1}X^{f}A=\left((A^{-1})^{f^{-1}}XA^{f^{-1}}\right)^{f}$, where
$A\in\text{GL}(d,q)$. This proves that we can move elements of
$\mathcal{B}$ to the right of the product of automorphisms. This
proves our theorem.
\end{proof}
Now notice that the order of $\mathcal{A}$ is actually big, it is
$q^{\frac{d(d-1)}{2}}(q^d-1)\cdots(q-2)$ but the order of
$\mathcal{B}$ is small. The group $\mathcal{B}$ is the direct product
of the graph and field automorphisms. The order of $\mathcal{B}$ is
$2\gamma$, where $\gamma$ is the
degree for the extension $\mathbb{F}_q$ over the prime
subfield. Let $\gamma_1=2\gamma$.

Let $\phi$ and $\phi^m$ be as in Section \ref{MOR}, then from the
previous theorem $\phi=A\psi_1$
and $\phi^m=A^\prime\psi_2$, where $A,A^\prime\in\mathcal{A}$ and
$\psi_1,\psi_2\in\mathcal{B}$. I shall
consider $A\in\mathcal{A}$ as the conjugator as well, this is clearly
the case because $\mathcal{A}\cong\text{PGL}(d,q)$.

Now if $\phi=A\psi_1$, then $\phi^m=AA^{\psi_1}\cdots
A^{\psi_1^{m-2}}A^{\psi_1^{m-1}}\psi_1^{m}$. In this case\\ $AA^{\psi_1}\cdots
A^{\psi_1^{m-2}}A^{\psi_1^{m-1}}\in\mathcal{A}$ and $\psi_1^m\in\mathcal{B}$.

Now if $\gamma_1<m$ and since the order of $\psi_1$ divides
$\gamma_1$, there are $r_1$ and $r_2$ such that $m-1=k_1\gamma_1+r_1$,
where $0\leq r_1<\gamma_1$ and $r_2=m\mod\gamma_1$. Then
$AA^{\psi_1}\cdots A^{\psi_1^{m-1}}\psi_1^{m}=A_1^{k_1}AA^{\psi_1}\cdots
A^{\psi_1^{r_1}}\psi_1^{r_2}$, where $A_1=AA^{\psi_1}\cdots
A^{\psi_1^{\gamma_1-1}}$. 
From the information of $\phi$ and
$\phi^m$ we then have the information of $\psi_1$ and
$\psi_1^{r_2}$. For all practical purposes of implementing this
cryptosystem, the degree of the field extension cannot be too large,
in this case one can do an exhaustive search on the cosets of
$\mathcal{A}$ and find out $\psi_1$ and $\psi_1^{r_2}$ and do another
exhaustive search to solve the discrete logarithm problem in $\psi_1$
and find the $r_2$. The information of $r_2$ gives us a vital
information about the secret key $m$. This is clearly unacceptable. So
the only way out of this situation is not to use automorphisms from
$\mathcal{B}$.

Then for
$X\in\text{SL}(d,q)$ the automorphisms $\phi$ and $\phi^m$ as in
Section \ref{MOR} is given by 
\begin{eqnarray}
\label{m1}\phi(X)&=A^{-1}XA&\text{for some}\;\;A\in\text{GL}(d,q)\\
\label{m2}\phi^m(X)&=A^{\prime{-1}}XA^\prime&\text{for
  some}\;\;A^\prime\in\text{GL}(d,q) 
\end{eqnarray}
Now notice, in the description of the MOR protocol, we presented
the automorphisms as action on generators and furthermore a set of 
generators for SL$(d,q)$ are the elementary transvections.   

In this case from the public information of $\phi$ and $\phi^m$ one
can find a candidate for $A$ and $A^\prime$. This problem is known to be easy in
GL$(d,q)$ and is often
refereed to as \emph{the special conjugacy
  problem}~\cite{paeng1,crypto2001}. 
However, notice that $A$ and $A^\prime$
are not unique. For example, if $A$ and $A^\prime$ satisfy the above equations then
so will $Az$ and $A^\prime z^\prime$ for any $z,z^\prime\in Z$, see
Section~\ref{think1}.

We just saw that the only way to build a secure MOR cryptosystem using
SL$(d,q)$ is using automorphisms from $\mathcal{A}$.
Henceforth, whenever we are talking about the MOR
cryptosystem, we are using the group
SL$(d,q)$ and the automorphisms from $\mathcal{A}$.
\section{Security of the proposed MOR cryptosystem}
This paper is primarily focused on the discrete logarithm
problem in the automorphism group of a non-abelian group. There are
two kinds of attack on the discrete logarithm problem over finite fields. One
is the generic attack; this attack uses a \emph{black box} group
algorithm and the other is an \emph{index calculus} attack.

Since the black box group algorithms work in any group, they will
work in the automorphism group too, see~\cite[Theorem
1]{asiacrypt2004}. We have no way to prevent that. On the other hand,
these generic attacks are of exponential time complexity and so is of
the least concern.

The biggest computational threat to any cryptosystem using the
discrete logarithm problem is a subexponential attack like the index
calculus attack~\cite{oliver}. It is often argued~\cite{koblitz,joseph}
that there is no 
index calculus algorithm for most elliptic curve cryptosystems
that has subexponential time complexity. This fact is often
presented to promote elliptic curve cryptosystem over a finite
field cryptosystem~\cite{koblitz}. So, the best we can hope from the
present MOR cryptosystem 
is that there is no index calculus attack or the index calculus attack
becomes exponential.
\subsection{Inner automorphisms as matrices}
As it turns out the only way that a secure MOR cryptosystem might work
for the unimodular group is through conjugation of matrices, i.e.,
automorphisms from $\mathcal{A}$. This MOR cryptosystem can be seen as
working with inner automorphisms of GL$(d,q)$. It is well known that the inner
automorphisms work linearly on the $d^2$-dimensional algebra of
matrices of degree $d$ over $\mathbb{F}_q$. For a fixed
basis, any
linear operator on a vector space can be represented as a matrix. 
So, the discrete logarithm
problem on $\langle\phi\rangle$ (Section \ref{MOR}) is now reduced to
the discrete logarithm problem in GL$(d^2,q)$.
The question is, how
easy is it to solve this discrete logarithm problem?

The best algorithm for solving the discrete logarithm problem in
GL$(k,q)$ was given by Menezes and Wu~\cite{menezes1}. In this case, the authors show that for 
$X,Y\in\text{GL}(k,q)$, such that, $X^l=Y$, $l\in\mathbb{N}$; we can
solve the discrete logarithm problem, if $\chi(x)$ the
characteristic polynomial of $X$, factors into irreducible polynomials
of small degree. If the characteristic polynomial is irreducible then
the discrete logarithm problem in $\langle X\rangle$ reduces to the
discrete logarithm problem in $\mathbb{F}_{q^k}$. 

In our case we are working in GL$(d^2,q)$. So the characteristic
polynomial has degree $d^2$. It is easy to see that if the characteristic polynomial is irreducible then the extension of
the lowest degree in which the characteristic polynomial will split is
$\mathbb{F}_{q^{d^2}}$. However this is not the case, since
$\phi(1)=1$, $1$ is an eigenvalue of $\phi$ and so the best we can
hope for is $\mathbb{F}_{q{d^2-1}}$.

\subsubsection{Recovering the conjugator up to a scalar multiple}\label{think1}
Let $\phi(X)=A^{-1}XA$, where $A\in\text{GL}(d,q)$. Since $\phi$ is
linear, if we take $X=1+e_{ij}$, $i\neq j$; then
$\phi(X)=A^{-1}XA=1+A^{-1}e_{ij}A$.
Now if we look at $e_{ij}A$ closely, then $e_{ij}A$ is a matrix where
the $j$\textsuperscript{th} row of $A$ is the
$i$\textsuperscript{th} row of $e_{ij}A$, and the rest all zeros. Since $A$ is
non-singular, all the contents of any row can not be all zeros. From this
it follows that the matrix $A^{-1}e_{ij}A$ consists of $d$ columns,
each of which is a constant multiple of the $i$\textsuperscript{th}
column of $A^{-1}$. One of these columns must be nonzero. Now 
consider $A^{-1}\left(1+e_{i,i+1}\right)A -1$ for $i=1,2,\ldots,d-1$,
and let each $I_i$ be a corresponding nonzero column. Then construct
a $d\times d$ matrix, whose first $d-1$ columns are the columns
$I_1,I_2,\ldots,I_{d-1}$ and the $d$\textsuperscript{th} column being
a nonzero column of $A^{-1}\left(1+e_{n1}\right)A -1$. Then we end up
with a matrix $N=A^{-1}D$, where $D=[w_1,w_2,\ldots,w_d]$ is a
diagonal matrix. Since $N$ is known, we have found $A^{-1}$ up to a diagonal matrix.

It is obvious that $N^{-1}\phi(X)N=D^{-1}XD$ and hence
$N^{-1}\left(1+e_{ij}\right)N-1=w_i^{-1}w_je_{ij}$. Then by taking
$j=1,2,\ldots,d$ and $i=1$, we can find
$w_2^{-1}w_1,w_3^{-1}w_1,\ldots,w_d^{-1}w_1$, and form the diagonal
matrix $D^\prime=[1,w_1w_2^{-1},\quad\newline w_1w_3^{-1},\ldots,w_1w_d^{-1}]$. It is easy to
see now that $ND^\prime$ is $A^{-1}w_1$ and we have found $A$ up to a
scalar multiple.

It is not hard to convince oneself that once $A$ is found up to a
scalar multiple from $\phi$, in most cases the discrete logarithm problem in $\phi$
turns out to be a discrete logarithm problem in $A$. If one recovers
$Ac_1$ and $A^mc_2$ from $\phi$ and $\phi^m$, where $c_1,c_2\in\mathbb{F}_q^\times$, then
compute $(Ac_1)^{q-1}=A^{q-1}$ and $(A^mc_2)^{q-1}=A^{m(q-1)}$ and
solve the corresponding discrete logarithm problem. From Menezes-Wu~\cite{menezes1}
it is clear that this discrete logarithm problem can have a worst case
complexity of that of a discrete logarithm problem in
$\mathbb{F}_{q^d}$.

We can stop this attack by taking $A$ to be of order $q-1$. But, if the eigenvalues of $A$ are $\mu_1,\mu_2,\ldots,\mu_d$,
then the eigenvalues of $A^m$ are $\mu_1^m,\mu_2^m,\ldots,\mu_d^m$. On
the other hand the eigenvalues of $cA$ are
$c\mu_1,c\mu_2,\ldots,c\mu_d$, $c\in\mathbb{F}_q$. When one recovers
$c_1A$ and $c_2A^m$, $c_1,c_2\in\mathbb{F}_q$ one recovers
$c_1\mu_1,c_1\mu_2,\ldots,c_1\mu_d$ and
$c_2\mu_1^m,c_2\mu_2^m,\ldots,c_2\mu_d^m$. Then one can compute
$\dfrac{\mu_i}{\mu_j}$ and $\left(\dfrac{\mu_i}{\mu_j}\right)^m$, by
  taking quotients. Notice that these quotients belong to
  $\mathbb{F}_{q^d}$. However since there is no unique way to order the
  eigenvalues, one might not be able to match a quotient with its power. Then we might have to deal with several quotients to get to
  the right $m$. However, for all practical applications the size of
  the matrix $d$ is small and so this search is not going to cost much;
  on top of that one can do this in parallel. So it is resonable to
  claim at this stage that the discrete logarithm problem in $\phi$ is
  reduced to a discrete logarithm problem in $\mathbb{F}_{q^d}$.

The expected asymptotic complexity of the index calculus
algorithm in $\mathbb{F}_{q^k}$ is 
$\exp{\left((c+o(1))(\log{q}^k)^\frac{1}{3}(\log\log{q}^k)^\frac{2}{3}\right)}$
, where $c$ is a constant, see~\cite{oliver} and~\cite[Section 4]{koblitz}.
If the degree of the extension, $k$, is greater than
$\log^2{q}$ then the asymptotic time complexity of the index calculus
algorithm is exponential. In our 
case that means, if $d>\log^2{q}$ then the asymptotic complexity
of the index calculus algorithm becomes exponential.

If we choose $d\geq\log^2{q}$ then this MOR cryptosystems becomes as
secure as the ElGamal over the elliptic curve groups, because the index
calculus algorithm is exponential; otherwise we can not
guarantee. But on the other hand in the proposed MOR cryptosystem
encryption and decryption 
works on $\mathbb{F}_q$ and breaking the cryptosystem depends on
solving a discrete logarithm problem on $\mathbb{F}_{q^d}$. Since,
implementing the index calculus attack becomes harder as the field
gets bigger, it is clear that if we take 
$d\ll\log^2{q}$, then the MOR cryptosystem is much more secure than the
ElGamal cryptosystem over $\mathbb{F}_q$.

\section{Implementation of this MOR cryptosystem}
The cryptosystem we have in mind is the MOR cryptosystem (Section
\ref{MOR}), the non-abelian group is SL$(d,q)$ and the automorphisms
are the automorphisms from $\mathcal{A}$. In this implementation the
most important thing will be the presentation of $\phi$ and
$\phi^m$. We decided earlier that the presentation will be the action
of the automorphisms on a set of generators
$\{g_1,g_2,\ldots,g_\tau\}$. Now we can write $\phi(g_i)$ as a word
in the generators $g_1,g_2,\ldots,g_\tau$ or we can write the
product of the generators as a matrix. We choose the later,
there are two reasons for that:
\begin{description}
\item This will contain the growth in the length of the word,
  especially while computing the powers of $\phi$. That will stop any
  length based attack.
\item This will add to the diffusion.
\end{description}
The set of generators for SL$(d,q)$ that we have in mind is the
elementary transvections. It is easy to go back and forth as words in
elementary transvections and matrices using row reduction.

A big question is how to compute large powers of $\phi$ efficiently? This is not
the principal object of study for this paper and we will be brief on this
topic.   

Since a set of generators are elementary transvections, computing the
power of $\phi$ can be done using only words in elementary transvections and the
image of the automorphism on these elementary transvections. This can be done very
efficiently. However, we have decided to write $\phi^m(g_i)$ as
matrices. So, while computing the power of $\phi$, one might have to
go back and forth between words and matrices. The objective of this
exercise is to reduce the amount of matrix 
multiplication, in computing the power of $\phi$. 
Also, one can use the 
relations among the elementary transvections to shorten the length of
the word. There are quite a few options available.

We explore one such option in more details. Assume that we are computing
the $\phi^m$ using the \emph{square and multiply}
algorithm~\cite[Algorithm 5.5]{stinson}. 
In this algorithm one needs to multiply
two group elements, in our case it is composing two automorphisms.
So, we need to find out the worst-case
complexity for multiplying two automorphisms. I further assume that the
automorphism is given as the image of $\left(1+e_{i,j}\right)$, $i\neq
j$, $i,j\in\{1,2,\ldots,d\}$, the
image is one $d\times d$ matrix. So for sake of notational convenience
I assume that we are squaring $\phi$, where $\phi$ is given by the
action on elementary transvections. As is customary we assume that the
field addition is free and we count the number of field multiplications
necessary to do the computation. 

Let's start with the matrix $M$ such that
$M=\phi\left(1+e_{i,j}\right)$, I shall use
row operations to write $M$ as product of elementary
transvections. We count each row operation as $d$ field multiplications
and there are utmost $d^2$ row operation. So in the worst case after
$d^3$ many field multiplication we have written $M$ as a product of
elementary transvection. At most there are $d^2$ many elementary
transvections in the product\footnote{Some small examples computed by
  the author using GAP \cite{GAP4} suggests that in practice this
  number is much smaller.}. 

Using the relation in Equation \ref{slrelations2}, we split each
transvection into product of elementary transvections over the prime
subfield. So now there are $\gamma d^2$ elementary transvections over the
prime subfield, for each of which the image under $\phi$ is known. Then
the image under $\phi$ is computed and then and then there are
$(p-1)\gamma d^4$ elementary transvection. The question is how to compute the matrix
corresponding to that? I propose the following:

There are utmost $(p-1)\gamma d^4$ elementary transvections in the product of
$\phi(M)$. Make $d$ equally spaced partition of the product of $\phi(M)$. Then each one of
these partitions can have utmost $(p-1)\gamma d^3$ elementary transvections. Now
we multiply the $(p-1)\gamma d^3$ elementary transvections to get $d$ many matrices
and them multiply these $d$ many matrices to get the final matrix
corresponding to $\phi^2\left(1+e_{i,j}\right)$. Now we multiply the
$(p-1)\gamma d^3$ elementary transvections linearly, one after the other, and use
the relations in Equations \ref{slrelations1} and \ref{slrelations2}. Notice that one of the components in
this multiplication is an elementary transvection. So every
multiplication can take utmost $d$ many field multiplication. So the
total complexity of multiplying $(p-1)\gamma d^3$ many elementary transvections is
$(p-1)\gamma d^4$. Since different partitions can be multiplied in parallel, we
assume that the worst-case complexity is $(p-1)\gamma d^4$ field multiplications.

Now we have to multiply the $d$ many matrices thus obtained. We
assume that we use a straight line program to compute the
product. Assuming that matrix multiplication can be done in $d^3$
field multiplication, we see that this also requires $d^4$ field
multiplications. Since we can compute $\phi^2\left(1+e_{i,j}\right)$ in
parallel for different $i$ and $j$, we claim that we can multiply
two automorphisms with worst-case complexity 
$(p-1)\gamma d^4+d^4$ field multiplications.
   
\subsection{Parameters for the cryptosystem}
We realized that if the conjugator $A$ in $\phi$ (Equation \ref{m1})
is a monomial matrix 
then that prevents the formation of a discrete logarithm problem in
the $\lambda$ of an elementary transvection $1+\lambda e_{i,j}$. We need
the inner automorphism so that the attack due to small cycle size of
the permutation in the monomial matrix can be avoided. So we have to
take the automorphism as conjugation by
$A\in\text{GL}(d,q)$. 

The size of $d$ and $q$ is an open question. With the limited amount of
knowledge we have about this cryptosystem, we can only make a
preliminary attempt to encourage further research. The current
standard for security in the public key cryptography is 80-bit
security. This means that the best known attack to the cryptosystem
should take at least $2^{80}$ steps.

The best known attack on the discrete logarithm problem in the
matrices $A$ and $A^\prime$ (Equations \ref{m1} and \ref{m2}) is
the generic \emph{square root} 
attack. So we have to ensure that to find $m$ from $A$ and $A^\prime$
one needs at least $2^{80}$ steps. For an attack algorithm we assume that computing in
$\mathbb{F}_q$ and in GL$(d,q)$ takes the same amount of time. If we
assume that the order of the matrix $A$ is the same as the order of
the field\footnote{The size of the field is
  motivated by the use of similar field in elliptic curve
  cryptography. For elliptic curves, the choice depends on the fact
  that the size of the 
  group of rational points on  an elliptic curve is roughly the size
  of the field. In our case, there are matrices of high order in
  GL$(d,q)$. So the field can be chosen smaller,
  depending on the matrix we choose to use.}, then the order of the
field should be around $2^{160}$.
So there are two 
choices for $q$, take $q$ to be a prime of the order
$2^{160}$, i.e., a 160 bit prime; or take $\mathbb{F}_q=\mathbb{F}_{2^{160}}$. 

A similar situation arises with the discrete logarithm problem over
the group of an elliptic curve over a finite field. The MOV attack reduces the discrete logarithm problem in the
group of the elliptic curve over $\mathbb{F}_q$ to a discrete logarithm problem in
$\mathbb{F}_{q^{k}}^{\times}$ for some positive integer $k$. This is of concern
in the implementation of the elliptic curve cryptosystem, because if
$k$ is too small then there is an subexponential attack on the
elliptic curve discrete logarithm problem. On the other hand, the size
of the elliptic curve group is almost as big as the field. To prevent the square root
attack the size of the field has to be considerably higher. Once you
assume that the field is of appropriate size $(2^{160})$, small $k$
provides adequate security. Our case is quite
similar. 

Koblitz et al.~\cite[Section 5.2]{koblitz} mentions that in practice
$k\approx 20$ is enough for security. If we buy
their argument, then it would seem that one can choose $d$ to be a
around $20$. We suspect that one might be able to go even
smaller. In our MOR cryptosystem, Menezes-Wu algorithm reduces the
discrete logarithm problem in $\mathbb{F}_{q^d}$.

So we propose $d=19$, and $q$ is as described earlier. Then we see
that if $q=2^{160}$, then we are talking about a discrete
logarithm problem in $\mathbb{F}_{2^{3040}}$. This
clearly surpasses every standard for discrete logarithm problem over
finite fields. At this size of the field, it does not matter if the
index-calculus is exponential or sub-exponential. It is simply not doable.    
\subsection{Generators for the cryptosystem}
The question I raise in this section is, are their better
generators than the elementary transvections in SL$(d,q)$? We saw that if we use
the elementary transvections for a prime field, then one needs
$(d^2-d)$ elementary transvections and $(d^2-d)\gamma$ elementary
transvections for $\mathbb{F}_q$ where $q=p^{\gamma}$. 

This is one of the major problems in the implementation of this
cryptosystem. 
We now try to solve this problem for SL$(d,p)$, where $p$ is a prime. In this MOR cryptosystem
(Section \ref{MOR}), generators play a major role. There are some
properties of the generators that help. Two of them are:
\begin{description}
\item[i] There should be an efficient algorithm to solve the word problem
  in these generators.
\item[ii] The less the number of generators of the group, the better is the
  cryptosystem.
\end{description}

Albert and Thompson~\cite{twogenerator} provides us with two
generators for SL$(d,q)$. They are
\begin{eqnarray*}
&\text{C}=1+\alpha e_{d-1,2}+e_{d,1}\\
&\text{D}=(-1)^d\left(e_{1,2}-e_{2,3}+\sum\limits_{i=3}^de_{i,i+1}\right)
\end{eqnarray*} 
where $\alpha$ is a primitive element of $\mathbb{F}_q$. It is clear
from the proof of~\cite[Lemma 1]{twogenerator} that to solve the word
problem in these generators one has to solve the discrete logarithm
problem in $\mathbb{F}_q$. This is clearly not useful for our cause. So
we adapt the generators and extend it to show that for
these generators one can compute the elementary transvections. Since
the number of generators is $2$, this gives us an
advantage for the presentation of the public key and the
ciphertext over elementary transvections. However, I know of no efficient algorithm to
solve the word problem in these generators. If we can find one such
algorithm then it can
be argued that this cryptosystem would become more economical(efficient).     

I now prove a theorem which is an adaptation of~\cite[Lemma
1]{twogenerator}. I use the convention used by Albert and Thomson,
\[e_{i,j}=e_{d+i,j}=e_{i,d+j}.\]
The proof of this lemma is practically identical with the proof
of~\cite[Lemma 1]{twogenerator}. I include a short proof for the
convenience of the 
reader and some of the formulas we produce in the proof are useful for
implementation.
\begin{theorem}
Let $$C=1+e_{d-1,2}+e_{d,1}\;\;\; \text{and}\;\;\;
D=(-1)^d\left(e_{1,2}-e_{2,3}+\sum\limits_{i=3}^de_{i,i+1}\right)$$ be
elements of SL$(d,p)$ where $d\geq 5$. Then $C$ and $D$ generates SL$(d,p)$. 
\end{theorem}
\begin{proof}
Let $G_0$ be the subgroup of SL$(d,p)$ generated by $C$ and $D$. I
will now write down a few formulas, which follow from direct
computation. For $2\leq k\leq d-2$ we have
\begin{eqnarray}
&D^{-1}=&(-1)^d\left(e_{2,1}-e_{3,2}+\sum\limits_{i=3}^de_{i+1,i}\right)\\
&C_1=&D^{-1}CD=1-e_{d,3}+e_{1,2}
\end{eqnarray}
\begin{eqnarray}
&CC_1C^{-1}C_1^{-1}=&1+e_{d,2}
\end{eqnarray}
\begin{eqnarray}
D^{k}=(-1)^{dk}\left(-e_{1,1+k}-e_{2,2+k}+\sum\limits_{i=3}^de_{i,i+k}\right)\\
D^{-k}=(-1)^{dk}\left(-e_{1+k,1}-e_{2+k,2}+\sum\limits_{i=3}^de_{i+k,i}\right)
\end{eqnarray}
\begin{eqnarray}
&C_{k}=D^{-k}CD^{k}=1-e_{k-1,k+2}-e_{k,k+1}\\
&C_{k}^{-1}=1+e_{k-1,k+2}+e_{k,k+1}\\
&\left(1+e_{d,k}\right)C_{k}\left(1-e_{d,k}\right)C_{k}^{-1}=1-e_{d,k+1}
\end{eqnarray}
From Equation (11) we see that $1+e_{d,2}$ belongs to $G_0$ and
then we use mathematical induction on $k$ and Equation (16) proves that $1+e_{d,k}\in G_0$ for
$k=2,\ldots,d-1$. Also $D^{-2}\left(1+e_{d,d-1}\right)D^2=1+e_{2,1}\in
G_0$. Furthermore $\left[1+e_{d,2},1+e_{2,1}\right]=1+e_{d,1}$. This
proves that $1+e_{d,k}\in G_0$ for $k=1,2,\ldots,d-1$. Then we can use the relations in SL$(d,p)$ to
prove that $1+e_{i,j}\in G_0$ for $i,j\in\{1,2,\ldots,d\}$ and $i\neq
j$. This proves the theorem.
\end{proof}
The proof of the theorem is constructive. It gives us a way to compute
the elementary transvections from these generators of Albert and Thomson; one can use
them effectively to publish the public key. There will be some
precomputation involved to change the action of $\phi$ from these
generators to elementary transvections.
\section{Conclusions} This paper studies the MOR cryptosystem for
the special linear group over finite fields. Cryptography is primarily
driven by applicability. So it is natural to ask, how efficiently can one implement 
this MOR cryptosystem? How secure is the cryptosystem? I talked in details on both
these issues in Sections 8 and 7 respectively.
These are often hard questions to answer from a preliminary investigation. The worst
case complexity is often far off from the actual cost of
computation and security in itself is a very elusive concept.
We now offer some realistic expectations on the computational
cost of this MOR cryptosystem when $q=2^\gamma$.

From the small experiments we did, it seems reasonable to assume that
a randomly chosen element of SL$(d,q)$ is generated by approximately $d$ elementary
transvections, not $d^2$ elementary transvections. This story is also
corroborated by the proof of the previous theorem, where we show that
SL$(d,p)$ is generated by all transvections of the form
$1+e_{d,k}$, $k=1,2,\ldots, d-1$ and by Humphries~\cite{humphries}.

Then we need to
compute the image of these $d$ elementary transvections under the
automorphism $\phi$. For that we
need to split each elementary transvections into product of elementary
transvections over the ground field using Equation
\ref{slrelations2}. Then in the worst case we now have $\gamma d$
elementary transvections. But since in any random binary string of
length $\gamma$ on average there are utmost $\dfrac{\gamma}{2}$
ones. So a more realistic expectation of the number of transvections
is $\dfrac{\gamma}{2}d$. Using the same expectation as before the image
of these transvections under $\phi$ will be a string of
$\dfrac{\gamma}{2}d^2$ elementary transvections. Now if we use a
straight line program, i.e., use the elementary transvections to
multiply the one next to it to form the matrix, then the worst case
complexity will be $\dfrac{\gamma}{2}d^3$ field
multiplication. However, in reality that complexity will be
something like $\dfrac{\gamma}{2}d^\lambda$ where $2<\lambda\leq
3$. So it is safe to assume that in practice $\lambda$ will be
around $2.5$.

With all this understanding we can say that if $q$ is a field of
characteristic $2$ and degree $\gamma$, then composition of two
automorphisms require around
\[d^2+\dfrac{\gamma}{2}d^{2.5}\]
field multiplications. If we were working with a finite field
$\mathbb{F}_{q^d}$, then the naive product of two non-zero field element costs around
$d^2$ field multiplications. We are quite close to that. 

Lastly, I recommend that the plaintext should be an elementary
transvection. It is known that trace and determinant is invariant
under matrix conjugation. So the trace or
the determinant can give out information about the plaintext. However,
if it is an elementary transvection, then the trace is always $d$ and
the determinant $1$.
\begin{small}
\bibliography{paper}

\providecommand{\bysame}{\leavevmode\hbox to3em{\hrulefill}\thinspace}
\providecommand{\MR}{\relax\ifhmode\unskip\space\fi MR }
\providecommand{\MRhref}[2]{%
  \href{http://www.ams.org/mathscinet-getitem?mr=#1}{#2}
}
\providecommand{\href}[2]{#2}
\begin{thebibliography}{10}

\bibitem{twogenerator}
A.A. Albert and John Thompson, \emph{Two-element genration of the projective
  unimodular group}, Illionois Journal of Mathematics \textbf{3} (1959),
  421--439.

\bibitem{alperin}
J.L. Alperin and Rowen~B. Bell, \emph{Groups and {R}epresentations}, Springer,
  1995.

\bibitem{carter}
Roger~W. Carter, \emph{Simple groups of {L}ie type}, John Willey \& Sons, 1989.

\bibitem{don}
Don Coppersmith and Shmuel Winograd, \emph{Matrix multiplication via arithmatic
  progression}, Proceedings of the nineteenth annual ACM conference on Theory
  of Computing, 1987, pp.~1--6.

\bibitem{jean}
Jean Dieudonn\'{e} and Loo-Keng Hua, \emph{On the automorphisms of the
  classical groups}, Memoirs of the American Mathematical Society (1951),
  no.~2.

\bibitem{GAP4}
The GAP~Group, \emph{{GAP -- Groups, Algorithms, and Programming, Version
  4.4.10}}, 2007.

\bibitem{humphries}
Stephen~P. Humphries, \emph{Generation of special linear groups by
  transvections}, Journal of Algebra \textbf{99} (1986), 480--495.

\bibitem{koblitz}
Neal Koblitz, Alfred Menezes, and Scott Vanstone, \emph{The state of elliptic
  curve cryptography}, Designs, Codes and Cryptogrpahy \textbf{19} (2000),
  173--193.

\bibitem{asiacrypt2004}
In-Sok Lee, Woo-Hwan Kim, Daesung Kwon, Sangil Nahm, Nam-Soek Kwak, and Yoo-Jin
  Baek, \emph{On the security of {MOR} public key cryptosystem}, Asiacrypt 2004
  (P.J.Lee, ed.), LNCS, no. 3329, Springer-Verlag, 2004, pp.~387--400.

\bibitem{ayan1}
Ayan Mahalanobis, \emph{A simple generalization of {ElGamal} cryptosystem to
  non-abelian groups}, Communication in Algebra \textbf{36} (2008), no.~10,
  3878--3889.

\bibitem{menezes1}
Alfred Menezes and Yi-Hong Wu, \emph{The discrete logarithm problem in
  {GL}$(n,q)$}, Ars Combinatorica \textbf{47} (1997), 23--32.

\bibitem{paeng1}
Seong-Hun Paeng, \emph{On the security of cryptosystem using the automorphism
  groups}, Information Processing Letters \textbf{88} (2003), 293--298.

\bibitem{crypto2001}
Seong-Hun Paeng, Kil-Chan Ha, Jae~Heon Kim, Seongtaek Chee, and Choonsik Park,
  \emph{New public key cryptosystem using finite non-abelian groups}, Crypto
  2001 (J.~Kilian, ed.), LNCS, vol. 2139, Springer-Verlag, 2001, pp.~470--485.

\bibitem{rotman}
Joseph~J. Rotman, \emph{An introduction to the theory of groups}, 4 ed.,
  Springer-Velag, 1994.

\bibitem{oliver}
Oliver Schirokauer, Damian Weber, and Thomas Denny, \emph{Discrete logarithm:
  the effectiveness of the index calculus method}, Algorithmic number theory
  (Talence, 1996), LNCS, vol. 1122, 1996, pp.~337--361.

\bibitem{joseph}
Joseph Silverman and Joe Suzuki, \emph{Elliptic curve discrete logarithms and
  the index calculus}, Asiacrypt'98 (K.~Ohra and D.~Pei, eds.), LNCS, vol.
  1514, 1998, pp.~110--125.

\bibitem{steinberg}
Robert Steinberg, \emph{Automorphisms of finite linear groups}, Canadian
  Journal of Mathematics \textbf{12} (1960), 606--615.

\bibitem{stinson}
Douglas Stinson, \emph{Cryptography theory and practice}, third ed., Chapman \&
  Hall/CRC, 2006.

\end{thebibliography}
\bibliographystyle{amsplain}
\end{small} 
\end{document}